\documentclass[smallextended]{svjour3}  

\usepackage[applemac]{inputenc}   
\usepackage[T1]{fontenc}
\usepackage{lmodern}
\usepackage{graphicx}  
\usepackage{mathptmx}  
\usepackage[centertags]{amsmath} 
\usepackage{amsfonts}
\usepackage{amssymb}
\usepackage{ushort}
\usepackage{type1cm}  
\usepackage{bbold}
\usepackage{latexsym}  
\usepackage{enumerate} 
\usepackage[a4paper,total={12cm,20cm},centering,includehead]{geometry} 
\usepackage{hypernat}  
\usepackage[bookmarksopen,pdfauthor={Augustin de Maere},hyperfootnotes=false,naturalnames=true,plainpages=false]{hyperref} 
\hypersetup{
    colorlinks,%
    citecolor=black,%
    filecolor=black,%
    linkcolor=black,%
    urlcolor=black
}

\newcommand{\charf}[1]{\mathbb{1}_{#1}}
\newcommand{\real}{\mathbb R}
\newcommand{\nat}{\mathbb N}
\newcommand{\ent}{ \mathbb Z}
\newcommand{\plan}{\mathbb Z^d}
\newcommand{\realspace}{\mathbb R^d}

\newcommand{\prob}{\mathbb P}

\newcommand{\dif}{\mathrm{d}}

\newcommand{\norm}[1]{\left\lvert #1\right\rvert}
\newcommand{\grandnorm}[1]{\biggl\vert #1\biggr\vert}
\newcommand{\dnorm}[1] {\left\lVert #1 \right\rVert}

\newcommand{\proj}[1] {\Pi_{#1}}

\newcommand{\muinv}[1]{\mu_{\mathrm{inv}}^{\scriptscriptstyle (\!#1\!)}}
\newcommand{\sleb}[1]{\delta^{\scriptscriptstyle (\!#1\!)}}
\newcommand{\brond}[1]{\mathcal{B}^{\scriptscriptstyle (\!#1\!)}}
\newcommand{\gammaplus}{\gamma^{\scriptscriptstyle(\!+\!)}}
\newcommand{\gammamoins}{\gamma^{\scriptscriptstyle(\!-\!)}}

\newcommand{\vect}[1]{\ushort{#1}}
\spnewtheorem*{corollarystar}{Corollary}{\bf}{\it} 

\title{Exponential Decay of Correlations for Strongly Coupled Toom Probabilistic Cellular Automata}

\author{Augustin de Maere \and Lise Ponselet}
\institute{A. de Maere \and L. Ponselet \at UniversitŽ catholique de Louvain, IRMP, Chemin du cyclotron 2, bte L7.01.03, B-1348 Louvain-la-Neuve, Belgium.\\Tel.: +32-10-473279\\Fax: +32-10-472414\\ \email{lise.ponselet@uclouvain.be}}
\date{}
\smartqed
\begin{document}

\maketitle

\begin{abstract}
We investigate the low-noise regime of a large class of probabilistic cellular automata, including the North-East-Center model of Toom. They are defined as stochastic perturbations of cellular automata belonging to the category of monotonic binary tessellations and possessing a property of erosion. We prove, for a set of initial conditions, exponential convergence of the induced processes toward an extremal invariant measure with a highly predominant spin value. We also show that this invariant measure presents exponential decay of correlations in space and in time and is therefore strongly mixing.
\keywords{Probabilistic Cellular Automata \and Phase Transition \and Toom Models \and Monotonic Binary Tessellations }
\end{abstract}

\section{Introduction}\label{secIntro}
Probabilistic cellular automata (PCA) are a type of models of multicomponent systems whose time evolution is driven by local interactions blurred with some noise. More precisely, they are discrete-time stochastic processes with the Markov property, made of lattices of components whose individual states take values in a finite set and are simultaneously updated at every time step. The transition rules involve interactions between neighbouring components.

The long-time limit of these processes has been the subject of many numerical and theoretical results in the last fifty years and a lot of questions remain open - see for instance the surveys \cite{To95,DoKrTo90}. Despite the apparent simplicity of their discrete configuration space and merely local interactions, PCA exhibit a variety of macroscopic tendencies. Moreover, due to their basic definition, these models provide a favorable field of study in order to give firm foundations to the growing knowledge about non-equilibrium phenomena. In particular, interest concentrates on the existence of phase transitions in this context of non-equilibrium statistical physics. For a non-negligible subset of the parameter set of certain PCA models, the limiting states of the processes can be non-unique and depend strongly on the initial conditions, thus providing examples of systems which keep remembering part of the data from their remote past \cite{BeGr85}.  

Here we consider the PCA resulting from small random perturbations of the deterministic cellular automata in the class of \textit{monotonic binary tessellations} (MBT) examined in \cite{To76,To80}. In an MBT, the \textit{spin} variables disposed on the sites of a $\plan$ lattice can take only two different values and are all simultaneously updated according to the same monotonic function of their neighboursÕ states. In \cite{To76,To80}, Toom gave a criterion for an MBT to be an \textit{eroder}, i.e. to erase in a finite time any finite island of impurities in a predominantly homogeneous configuration of one of the two types of spins. The completely homogeneous space-time configuration with this spin value everywhere is then said to be an \textit{attractive} trajectory of that MBT. He also proved in \cite{To80} that this erosion condition implies the \textit{stability} of the homogeneous trajectory under the introduction of a small error rate: the corresponding PCA admits a stationary measure for which the low probability of finding an impurity is preserved. Another proof of this stability theorem in a more general context was later given in \cite{BrGr91}, using renormalization group methods.

For example, the one-dimensional percolation PCA of Stavskaya and Toom's two-dimensional North-East-Center voting PCA fall within this class of perturbations of MBT which are eroders. They both present a phase transition as non-uniqueness of their invariant probability measures in the low-noise regime follows from this stability theorem, while, when the noise is high enough, the processes converge toward a unique invariant measure, regardless of the initial state \cite{Do71}.

Beside this stability result, very little is known about the long-time asymptotics of this class of PCA in the nearly deterministic regime. It has already been explored by means of simulations \cite{BeGr85,DiMa11,Mak98,Mak99,VaPePi69}. Here we address the problem from a theoretical point of view. By virtue of the ergodic theorem, the asymptotic behavior of such a PCA is actually given by its ergodic invariant probability measures. We focus on the statistical properties of the invariant measure with a predominant spin value which is obtained as a perturbation of the stable homogeneous trajectory. Reference \cite{BeKrMa93} established exponential convergence to this stationary measure, and therefore exponential decay of its correlations, for the Stavskaya PCA and a class of multidimensional generalizations of it.

In this paper, we use a different method allowing us to extend this result to the North-East-Center model and all the other MBT satisfying ToomÕs erosion condition, with a restriction to MBT without memory, where the state of the system at any time $t$ depends only on its previous state at time $t-1$ and not directly on earlier states. We prove rigorously that for some well-chosen initial conditions, the process converges exponentially fast toward the invariant measure under consideration. We also show that this measure has exponential decay of correlations in space and in time. In particular, it is ergodic and therefore extremal. The proof is based on a perturbative expansion, with paths and graphs, which combines a technique of decoupling in the pure phases previously introduced and developed for coupled map lattices \cite{KeLi06,deM10} with graphs constructed for a Peierls argument in the proof of stability in \cite{To80}.

\section{Formalism and Results}\label{secMain}
We consider the cellular automata in the class of monotonic binary tessellations described in Section \MakeUppercase{\romannumeral 4} of \cite{To80} except that we restrict ourselves to MBT without memory. Some of our notations are similar to the notations of \cite{To80} in order to facilitate the use of the graphs constructed there but they are not identical because we don't need as general a framework here.

Let $S = \{-1,1\}$ be the binary state space for a component of the system, which we will call a \textit{spin}, and let $X = S^{\plan}$ be the configuration space for the whole system. For any spin configuration $\vect{\omega} \in X$, $\omega_x$ will denote the value of $\vect{\omega}$ at site $x$ and $\vect{\omega}_{\Lambda}$ will denote the vector $\{\omega_x\}_{x \in \Lambda}$ for any $\Lambda \subseteq \plan$. We will also use the notation $(\vect{\omega}_{\neq x}, a)$ for the configuration obtained from $\vect{\omega}$ by replacing the spin $\omega_x$ at site $x$ with the value $a\in S$.

In a tessellation, the evolution law for the state of the spin at a site $x$ of the $\plan$ lattice involves its neighbours, which are defined as the elements of $U(x) = x + U$ for a fixed finite set $U= \{u_1, \dotsc , u_R \} \subset \plan$. Starting from a configuration $\vect{\omega} \in X$, the simultaneous updating of every spin at every time step consists in transforming the spin $\omega_x$ at site $x$ into the value
\begin{equation*}
\phi_x(\vect{\omega})=\phi \left(\omega_{x+u_1}, \dotsc , \omega_{x+u_R}\right)
\end{equation*}
for a given function $\phi : S^R \to S$. From now on, as we consider monotonic binary tessellations, $\phi$ will be monotonic in the sense that if $\omega_u \leq \omega'_u$ for all $u \in U$, then $\phi(\vect{\omega}_U) \leq \phi({\vect{\omega}}_U')$. We also reject the trivial case of a constant function $\phi$. Note that these two assumptions imply that $\phi(\pm 1, \dotsc, \pm 1)=\pm 1$ and then that the configurations $\{+1\}_{x \in \plan}$ and $\{-1\}_{x \in \plan}$ are left invariant by the deterministic time evolution. They generate completely homogeneous trajectories.

In addition, we suppose that $\phi$ verifies the Toom erosion criterion. It is expressed in terms of the \textit{plus sets}: the subsets $Z \subseteq U$ such that if $\omega_u =+1$ for all $u \in Z$, then $\phi(\vect\omega_U)=+1$. When applied to the particular case of an MBT without memory, the criterion given in Proposition 1 of \cite{To76} and in Theorem 6 of \cite{To80} states that the intersection of all convex hulls of plus sets in $\realspace$ must be empty.

We now describe a way to introduce some noise in such an eroding MBT to generate a PCA. The system follows the same updating rule as in the deterministic case but, at each site of the lattice and at each time step, an error can occur with a probability of order $\epsilon$, for some $\epsilon$ in $[0,1]$, the spin then taking the opposite value. Occurrence of an error at a site is independent from occurrence of errors at other sites. The process resulting from the sequence of simultaneous updatings of all spins becomes a stochastic process. The local transition probabilities $p(\xi_x| \vect{\omega})$ of that process only depend on $\vect\omega_{U(x)}$. We suppose that they satisfy the following assumptions, for some given values of the parameters $\epsilon \in [0,1]$ and $\alpha\geq0$:
\begin{enumerate}[\hspace{1cm}({A}1)]
\item if $\xi_x \neq \phi_x (\vect{\omega})$, then $p(\xi_x| \vect{\omega}) \leq \epsilon$;\label{assump:2}
\item if $a=\phi_x (\vect{\omega})$, then $|p(\xi_x| \vect{\omega})-p(\xi_x| \vect{\omega}_{\neq y},a)| \leq \alpha \, p(\xi_x| \vect{\omega})$ for all $y \in \plan$.\label{assump:3}
\end{enumerate}
Assumption (A\ref{assump:2}) is the low-noise condition: it ensures that, at each site of the lattice, the deterministic rule is followed with a probability of at least $1-\epsilon$. Note that we make no restriction about the possible bias in favor of errors of a specific spin sign. Assumption (A\ref{assump:3}), used in conjunction with the monotonicity of $\phi$, ensures a form of decoupling in the pure phases: if a spin was flipped without changing the deterministic prescription given by the value of $\phi_x$, it would be of little consequence to the involved transition probabilities.  The local transition probabilities define a transition probability kernel on $X$, which can be formally written as
\begin{equation}\label{re:6}
\prob[\vect{\xi} | \vect{\omega}] = \prod_{x \in \plan} p(\xi_x| \vect{\omega}).
\end{equation}

Assume that $S$ is endowed with the discrete topology and $X$ with the product topology. Let $\mathcal{C}(X)$ be the set of continuous functions from $X$ into $\real$ with the norm: $\norm{\varphi}_{\infty}  = \sup_{\vect{\omega} \in X} \norm{\varphi(\vect{\omega})}$. Since $X$ is compact, $\mathcal{C}(X)$ with the norm $\norm{.}_{\infty}$ is a Banach space, and its dual is $\mathcal{M}(X)$, the space of signed Borel measures on $X$ with the norm:
\begin{equation*}
\norm{\mu} = \sup \{\, \mu(\varphi) \mid \varphi \in \mathcal{C}(X),\, \norm{\varphi}_{\infty} \leq 1\, \}.\label{re:2}
\end{equation*}
For PCA on infinite spaces, it is often useful to consider the following semi-norm on $\mathcal{C}(X)$: if $\delta_x \varphi(\vect{\omega})  = \varphi(\vect{\omega})-\varphi(\vect{\omega}_{\neq x},- \omega_x)$, we define:
\begin{align}
\dnorm{\varphi}  = \sum_{x \in \plan} \norm{\delta_x \varphi}_{\infty}.\label{re:1}
\end{align}

Let $T: \mathcal{M}(X) \to \mathcal{M}(X)$ be the transfer operator associated to the transition probability kernel of \eqref{re:6}, explicitly defined as:
\begin{equation}
T \mu(\varphi) = \int \dif \vect{\omega} \int \dif \vect{\xi}\; \varphi(\vect{\xi})\, \prob[\vect{\xi}|\vect{\omega}]\, \mu(\vect{\omega}),\label{re:3}
\end{equation}
for any $\varphi$ in $\mathcal{C}(X)$.

Let $\charf{-,\Lambda}$ denote the indicator function of the subset $\{ \vect{\omega} \in X | \omega_x  = -1\ \forall x \in \Lambda\}$ and, for any measurable set $Y \subseteq X$, let $\charf{Y}: \mathcal{M}(X) \to \mathcal{M}(X)$ be the operator defined as:
\begin{equation*}
\charf{Y}\mu(\varphi) =  \mu(\varphi\, \charf{Y} ) \quad \forall \varphi \in \mathcal{C}(X).\label{re:5}
\end{equation*}

As the MBT defined above satisfies the Toom erosion condition, the results of \cite{To80} apply to it. Theorems 5 and 6 of \cite{To80} imply the stability of the homogeneous trajectory with the spin value $+1$ everywhere. More precisely, let $\sleb{+}$ be the Dirac measure concentrated on the configuration $\{+1\}_{x \in \plan}$. Then, using these theorems and Definition 1 of a stable trajectory in \cite{To80}, we have
\begin{equation}\label{stabil}
\lim_{\epsilon \to 0} \sup_{\substack{n \in \nat \\ x\in \plan}} T^n\sleb{+}  \left(\charf{-,\{x\}}\right) = 0 .
\end{equation}
Now consider the following sequence of measures:
\begin{equation*}
\left( \frac{1}{n} \sum_{k= 0}^{n-1} T^k \sleb{+}\right)_{n\in \nat_0}.
\end{equation*}
We can always extract from it a weakly-* convergent subsequence. Indeed, the Banach-Alaoglu theorem states that the unit ball of $\mathcal{M}(X)$ is compact in the weak-* topology. The associated limit $\muinv{+}$ is explicitly given by
\begin{equation}\label{muinv}
\muinv{+} (\varphi) = \lim_{j \to \infty} \left( \frac{1}{n_j} \sum_{k=0}^{n_j-1} T^k \sleb{+} (\varphi) \right) \quad \forall \varphi \in \mathcal{C}(X),
\end{equation}
for a certain subsequence $\left( n_j \right) _{j \in \nat}$ of increasing positive integers. The measure $\muinv{+}$ is obviously an invariant probability measure, in the sense that $T \muinv{+} = \muinv{+}$, and equation \eqref{stabil} implies that
\begin{equation}\label{muinvstabil}
\lim_{\epsilon \to 0} \sup_{x\in \plan} \muinv{+}  \left(\charf{-,\{x\}}\right) = 0 .
\end{equation}
\begin{sloppypar}\begin{example}[The Stavskaya model] 
The percolation PCA of Stavskaya \cite{StPi68} is one of the first PCA for which the existence of a phase transition has been rigorously proven - see an outline of the proof using the method of contours in Chapter 1 of \cite{DoKrTo90}. It has $d=1$, $U=\{ 0,1\} \subset \ent$ and $\phi (\omega_0,\omega_1) = -1$ if and only if $\omega_0=\omega_1=-1$. Errors are strongly biased in the sense that they can only transform a spin $+1$ into a $-1$ but not a $-1$ into a $+1$. The Dirac measure $\sleb{-}$ is then always a stationary state. We will call it $\muinv{-}$. Besides, the Toom criterion is satisfied by these $U$ and $\phi$ since the intersection of the plus sets $\{0\}$ and $\{1\}$ is empty, therefore equation \eqref{muinvstabil} carries the existence of an $\epsilon_c > 0$ such that for all $\epsilon \leq \epsilon_c$, the process admits a second different invariant measure $\muinv{+}$.
\end{example}
\begin{example}[The Toom model]
The North-East-Center voting model has $d=2$ and $U=\{(0,0);(1,0);(0,1)\} \subset \mathbb Z ^2$ and $\phi$ returns the majority spin among these three neighbours' states. This model verifies the Toom erosion condition too, because $\{(0,0);(1,0)\}$, $\{(1,0);(0,1)\}$ and $\{(0,1);(0,0)\}$ are plus sets whose convex hulls have an empty intersection. Moreover, the model is symmetric under the interchange of positive and negative spin signs. Combining this symmetry with definition \eqref{muinv} gives two invariant measures, $\muinv{+}$ which satisfies equation \eqref{muinvstabil} and $\muinv{-}$ which satisfies the symmetric counterpart of equation \eqref{muinvstabil}, thus revealing that they differ as long as $\epsilon$ is small enough.
\end{example}\end{sloppypar}

Then, in such models, if $\epsilon$ is below a critical threshold $\epsilon_c $, there exists an infinite number of invariant probability measures, since any convex combination of $\muinv{+}$ and $\muinv{-}$ is also an invariant probability measure. But, among these invariant measures, some have interesting properties such as ergodicity and exponential decay of correlations.
In \cite{BeKrMa93}, the low-noise regime of a class of PCA is examined, including the Stavskaya model and a multidimensional generalization of it, but not the Toom model. Exponential convergence toward equilibrium of the processes with the initial condition $\sleb{+}$ is proven for these PCA by constructing a cluster expansion. Exponential decay of correlations in space and in time follows for the non-trivial natural invariant measure $\muinv{+}$ that appears at the phase transition of the Stavskaya model.

Our argument extends the results of \cite{BeKrMa93} to the whole class of PCA defined above. It relies on an expansion which isolates the influence of each space-time site on each other site in its future, along several paths of influence. These paths pass through the dominant phase with positive spins almost everywhere and Assumption (A\ref{assump:3}) of decoupling in the pure phases can be used in order to bound the influence of spins one on another. The paths seldom encounter a negative spin. Whenever they do, in view to evaluate how improbable that negative spin is, we will attach to it a graph borrowed from a Peierls argument in the proof of stability in \cite{To80} and bring Assumption (A\ref{assump:2}) into play. The fact that paths select one site at a time will allow us to make use of the one-dimensional Toom graphs by associating them to single sites separately. As for the cluster expansion in \cite{BeKrMa93}, it requires contours which can enclose blocks of several spins, as the Stavskaya contours do, but not the most general Toom graphs, as far as we can see.

In Sections \ref{section:path} to \ref{section:decouplpure}, we will show that any initial probability measure in a suitable basin of attraction of $\mathcal{M}(X)$ converges exponentially fast toward $\muinv{+}$.
We will consider the sets
\begin{equation*}
\brond{+}(K,\epsilon') = \Big\{\ \mu \in \mathcal{M}(X)\ \Big|\  \norm{\,\charf{-,\Lambda}\mu\,} \leq K\, \epsilon'^{\norm{\Lambda}}\quad \forall \Lambda \subseteq \plan\Big\} ,
\end{equation*}
with $K \geq 0$ and $\epsilon' \in [0,1]$ and prove the following result:
\begin{theorem} \label{thm:exp}
For any MBT without memory characterized by a non-constant monotonic function $\phi$ verifying the Toom erosion condition, there exist $\alpha_* >0$ and a positive and strictly decreasing function $\epsilon_*$ defined on $[0,\alpha_*[$ such that, for all $\alpha \in [0, \alpha_* [$ and $\epsilon \in [0, \epsilon_* (\alpha)[$, the following assertion is true for any PCA satisfying the corresponding Assumptions (A\ref{assump:2}) and (A\ref{assump:3}). For any probability measure $\mu$ in $\brond{+}(K,\epsilon')$ with $K \geq 0$ and $\epsilon' < \epsilon_* (\alpha)$, there exist some constants $C < \infty$ and $\sigma <1$ such that, for all $\varphi \in \mathcal{C}(X)$ and all $n \in \nat$,
\begin{equation*}
\norm{\,T^n \mu(\varphi) - \muinv{+}(\varphi)\,} \leq C\,\dnorm{\varphi}\,\sigma^n.
\end{equation*}
\end{theorem}
Note that for PCA which are symmetric under the interchange of the two types of spins, such as the North-East-Center model, the symmetric result for $\mu$ in the class $\brond{-} \left( K, \epsilon' \right) $ and for $\muinv{-}$ is also valid.

Eventually, in Section \ref{section:corollaries}, we will prove that $\muinv{+}$ has exponential decay of correlations in space and in time and is, consequently, strongly mixing.
\begin{corollarystar}\label{coroll:1}
Assume that $\alpha < \alpha_*$ and that $\epsilon < \epsilon_*(\alpha)$, with $\alpha_*$ and $\epsilon_*$ as given by Theorem \ref{thm:exp}. Then there exist some constants $C'< \infty$ and $\eta <1$ such that for any $\varphi, \psi$ in $\mathcal{C}(X)$ with $\dnorm{\varphi} < \infty$ and $\dnorm{\psi}< \infty$, and with a positive Manhattan distance $d(\varphi,\psi) $ between their supports, we have
\begin{equation*}
\norm{\muinv{+}(\varphi \psi) - \muinv{+}(\varphi) \muinv{+}(\psi)} \leq C' \left( \dnorm{\varphi} \norm{\psi}_{\infty} + \norm{\varphi}_{\infty} \dnorm{\psi} \right)   \eta^{d(\varphi,\psi)}.
\end{equation*}
\end{corollarystar}
For the exponential decay of correlations in time,
we define the operator $T: \mathcal{C}(X) \to \mathcal{C}(X)$:
\begin{equation*}
T \varphi(\vect{\omega}) = \int \dif \vect{\xi} \; \varphi(\vect{\xi}) \ \prob[\vect{\xi}|\vect{\omega}]\label{m:13},
\end{equation*}
which is simply the dual of the transfer operator acting on signed measures.
\begin{theorem}Assume that $\alpha < \alpha_*$ and that $\epsilon < \epsilon_*(\alpha)$ with $\alpha_*$ and $\epsilon_*$ as given by Theorem \ref{thm:exp}. Then, for any continuous functions $\varphi,\psi$ with finite supports, there exist some constants $C_{\varphi,\psi} < \infty$ and $\sigma < 1$ such that, for all $n \in \nat$,\label{coroll:2}
\begin{equation*}
\norm{\muinv{+}(\varphi T^n\psi) - \muinv{+}(\varphi) \muinv{+}(\psi)} \leq C_{\varphi,\psi}\, \sigma^{n}.
\end{equation*}
\end{theorem}

\section{Path Expansion} \label{section:path}
In this section, we will introduce a path expansion which is essentially equivalent to the Dobrushin criterion \cite{Do71}, using here a formalism which was originally introduced by Keller and Liverani for Coupled Map Lattices \cite{KeLi06}.

Let $\prec$ be any well-ordering of $\plan$.
The operator $\proj{x}: \mathcal{C}(X) \to \mathcal{C}(X)$ is defined as:
\begin{equation*}\label{d:2}
\proj{x}\varphi(\vect{\omega}) = \varphi(\vect{\omega}_{\succeq x}, \vect{a}_{\prec x}) -  \varphi(\vect{\omega}_{\succ x}, \vect{a}_{\preceq x}),
\end{equation*}
where, from now on,  $\vect{a}$ will denote the configuration for which $a_x = a =+1 $ for all $ x \in \plan$. $(\vect{\omega}_{\succeq x}, \vect{a}_{\prec x})$ is the configuration obtained from $\vect{\omega}$ by replacing the spins at all sites $y \prec x$ with the value $a$. Using telescopic sums, we can check that, for any continuous function $\varphi$,
\begin{equation}\label{d:3}
\varphi(\vect{\omega}) = \varphi(\vect{a}) + \sum_{x \in \plan} \proj{x} \varphi(\vect{\omega}).
\end{equation}
With a slight abuse of notation, let us denote by $\proj{x}: \mathcal{M}(X) \to \mathcal{M}(X)$ the dual of the operator $\proj{x}: \mathcal{C}(X) \to \mathcal{C}(X)$. The image of $\mathcal{M}(X)$ under this operator is actually included in the set
\begin{equation}\label{d:4}
\mathcal{M}_x = \{ \mu \in \mathcal{M}(X)\ |\ \mu(\varphi) = 0\ \textrm{if}\ \varphi(\vect{\omega})\ \textrm{is independent of}\ \omega_x \}
\end{equation}whose elements verify the following property:
\begin{equation}
\mu \in \mathcal{M}_x \quad    \Rightarrow \quad \norm{\mu(\varphi)} \leq \norm{\mu} \norm{\delta_x \varphi}_{\infty}.\label{d:7}
\end{equation}

With equation \eqref{d:3}, we can see that any signed measure of zero mass $\mu \in \mathcal{M}(X)$ with $\mu(1)=0$ admits the following decomposition:
\begin{equation}\label{d:8}
 \mu = \sum_{x \in \plan} \proj{x} \mu.
\end{equation}
While $\proj{x} \mu$ belongs to $\mathcal{M}_x$, it is no longer the case for $T \proj{x} \mu$.
Nevertheless, since the interactions are local, we will see that $T \proj{x} \mu$ can be expressed as the sum of $R$ signed measures: a first one in $\mathcal{M}_{x-u_1}$, a second one in $\mathcal{M}_{x-u_2}$, ... , and a last one in $\mathcal{M}_{x-u_R}$.
For this, consider an arbitrary measure $\mu_x \in \mathcal{M}_x$. Using definitions \eqref{re:3} and \eqref{d:4} of $T$ and $\mathcal{M}_x$, together with our hypothesis that $p(\xi_y| \vect{\omega})$ only depends on $\vect\omega_{y+U}$, it is easy to check that
\begin{equation} \label{path:1}
T \mu_x = \sum_{i=1}^R T^{(x-u_i,x)} \mu_x ,
\end{equation}
where $R$ new operators have been defined:
\begin{equation}\label{path:2}
T^{(x-u_i,x)} \mu (\varphi)  = \int \dif \vect{\omega} \int \dif \vect{\xi}\; \varphi(\vect{\xi}) \, \tau^{(x-u_i,x)} (\vect{\xi},\vect{\omega}) \, \mu(\vect{\omega})  \quad \text{for $i=1, \dotsc, R$}
\end{equation}
with the kernels
\begin{align}
\tau^{(x-u_i,x)} (\vect{\xi},\vect{\omega})&=  \left( \prod_{y \notin x-U} p(\xi_y| \vect{\omega}) \right) \left( \prod_{j=1}^{i-1} p(\xi_{x-u_j}| \vect{\omega}_{\neq x}, a) \right) \cdot \notag \\
& \quad \quad \cdot \Big( p(\xi_{x-u_i}| \vect{\omega}) - p(\xi_{x-u_i}| \vect{\omega}_{\neq x}, a) \Big) \left( \prod_{j=i+1}^{R} p(\xi_{x-u_j}| \vect{\omega}) \right).\label{path:6}
\end{align}
We notice that, for all $i$, the image of $\mathcal{M}(X)$ under the operator $T^{(x-u_i,x)}$ is included in $ \mathcal{M}_{x-u_i}$.

Let now $\mu$ be a signed measure of zero mass and consider $T^n \mu$. Using the decomposition \eqref{d:8} and applying \eqref{path:1} iteratively, we find
\begin{equation*}
T^n \mu= \sum_{x_0 \in \plan} \sum_{x_0-x_1 \in U } \cdots \sum_{x_{n-1}-x_n \in U } T^{(x_n ,x_{n-1})}\cdots T^{(x_1, x_{0})}\proj{x_0} \mu . \label{d:13}
\end{equation*}
This sum can be rewritten as a sum over paths. Indeed, if we introduce
\begin{equation*}
P_x=\left\{ \gamma: \{0,\ldots,n\} \to \plan \mid \gamma_n = x\ \text{and}\ \gamma_{t-1} - \gamma_t \in U \  \forall t \right\},
\end{equation*}
it is equivalent to the following compact expression:
\begin{equation}
T^n \mu= \sum_{x \in \plan}\ \sum_{ \gamma \in P_x}\ \prod_{t = 1}^{n} T^{(\gamma_t, \gamma_{t-1})} \proj{\gamma_0} \mu , \label{path:4}
\end{equation}
where the operators $T^{(\gamma_t, \gamma_{t-1})}$ have to be applied in chronological order.

In the case of a weakly interacting system, this sum over paths can be used to prove the existence of a unique invariant probability measure, under the assumptions of the Dobrushin criterion \cite{Do71}. The system we are considering here is certainly not weakly interacting. However, in order to prove Theorem \ref{thm:exp}, the idea will be to take advantage of the decoupling in the pure phases conveyed in Assumption (A\ref{assump:3}). This assumption will provide upper bounds of order $\alpha$ on the couplings $T^{(\gamma_t, \gamma_{t-1})}$ in the positive phase. Indeed, as $a=+1$ in \eqref{path:6}, those bounds will be obtained for the instants $t$ such that
the considered space-time configuration presents at time $t-1$ the value $+1$ for $\phi(\vect\omega_{\gamma_t+U})$. As for the negative phase, it will be shown to be infrequent enough, for a suitable choice of initial condition.
\section{Pure Phase Expansion} \label{section:pure}
For fixed $x \in \plan$ and $\gamma \in P_x$, if $\gamma$ denotes both the function and its trajectory \linebreak $\left\{ \left( t,\gamma_t  \right) \mid t=0, \ldots , n \right\}$, let us partition the trajectory into $\gammaplus \subseteq \gamma \setminus \{ \gamma_0 \}$ and $\gammamoins = (\gamma \setminus \{ \gamma_0 \} ) \setminus \gammaplus$. We define, for $t =0, \ldots , n $, the sets $F\left( \gammaplus , t \right) \subseteq X$ with the following indicator functions:
\begin{equation*}
\charf{F\left( \gammaplus , t \right)} (\vect{\omega}) =  \prod_{x:(t+1,x) \in \gammaplus} \charf{+} \left( \phi _x (\vect{\omega} ) \right) \prod_{x:(t+1,x) \in \gammamoins} \charf{-} \left( \phi_x (\vect{\omega} ) \right).
\end{equation*}
These subsets lead to a partition of the configuration space $X^{\{0,\dotsc , n \}}$. 
Inserting this partition in \eqref{path:4},
\begin{equation*}
T^n \mu = \sum_{x \in \plan}\ \sum_{ \gamma \in P_x} \sum_{\gammaplus \subseteq \gamma \setminus \{ \gamma_0 \}}  \prod_{t=1}^n \left( \charf{F(\gammaplus,t)}   T^{(\gamma_t,\gamma_{t-1})}  \right) \charf{F(\gammaplus,0)}  \proj{\gamma_0} \mu , \label{decoupling:1}
\end{equation*}
where $\charf{F(\gammaplus,n)}$ is nothing but the identity operator. Using property \eqref{d:7} together with the fact that the image of $\mathcal{M}(X)$ under $T^{(x,\gamma_{n-1})}$ is included in $\mathcal{M}_x$, we have, for all $\varphi \in \mathcal{C}(X)$,
\begin{equation}
\norm{T^n \mu(\varphi)}  \label{decoupling:2} \leq \sum_{x\in \plan}\ \sum_{ \gamma \in P_x} \sum_{\gammaplus \subseteq \gamma \setminus \{ \gamma_0 \}}  \norm{ \prod_{t=1}^n \left( \charf{F(\gammaplus,t)}   T^{(\gamma_t,\gamma_{t-1})}  \right) \charf{F(\gammaplus,0)}  \proj{\gamma_0} \mu} \norm{\delta_x \varphi}_{\infty} .
\end{equation}

Next, for given $x \in \plan$, $\gamma \in P_x$ and $\gammaplus \subseteq \gamma \setminus \{ \gamma_0\}$, we define the set
\begin{equation*}
\mathcal{E} \left( \gammaplus \right) = \left\{ \left( \vect{\omega}_t \right) _{t=0}^n \in X^{\{0,\ldots , n \}} \mid \vect{\omega}_{t} \in F \left( \gammaplus, t \right) \quad \forall t \in \{ 0, \dots, n \} \right\}.
\end{equation*}
We will now build a graph using the construction in \cite{To80} in order to control the extent of the negative phase. We will then rewrite the expansion \eqref{decoupling:2} as a cluster expansion in terms of a combination of paths and graphs.
\section{Graphs} \label{section:graphs}
As we included the Toom erosion condition in the hypotheses of Theorem 1, we know that the intersection of all convex hulls of plus sets of the considered MBT is empty. Using Lemma 12 of \cite{To80}, it implies the existence of an integer $q\in\nat_0$, of $q$ linear functions $L_1, \dotsc, L_q$ from $\real^{d+1}$ to $\real$ and of a positive $r$ such that the hypotheses of Theorem 1 of \cite{To80} are verified by the MBT and by the homogeneous space-time configuration with positive spins everywhere. In particular, we can use the structures which are defined in the Peierls argument proving this stability theorem in Section \MakeUppercase{\romannumeral 2} of \cite{To80}. We will not rewrite the whole construction here but only extract the objects of use for our purpose and their properties. We will also adapt the general formalism to the more restricted class of models under consideration here. We refer the reader to \cite{To80} for all details and, for an introduction to that very general proof, to the review given in Appendix A of \cite{LeMaSp90} in the particular case of the North-East-Center model. A good exercise could also be to consider the case of the Stavskaya model. One then obtains the contours described in Chapter 1 of \cite{DoKrTo90}.
   
In order to bound the probability of finding a negative spin at a given time $t_-$ and at a given site $x_-$ of the lattice, when the initial measure is $\sleb{+}$, \cite{To80} associates a structure called \textit{truss} to each space-time configuration presenting this negative spin. Its iterative construction is given in Section \MakeUppercase{\romannumeral 2}, Part $3$ of \cite{To80} where it is denoted by $\Pi$. This truss consists of an unordered sequence of \textit{edgers}, which are maps from $\{1, \dotsc, q \}$ to $ \{0, \ldots , t_- \} \times \plan$ whose image contains exactly two points. Here we use, much like in the review of \cite{LeMaSp90}, an equivalent reformulation. We view an edger as an unoriented edge that links two points of $ \{0, \ldots , t_- \} \times \plan$ and bears an extra attribute, namely a partition of the set $\{1, \dotsc, q \}$ of \textit{poles} between its two vertices. This natural bijection between edgers and edges extends to trusses which we treat as graphs on $ \{0, \ldots , t_- \} \times \plan$ composed of these edges corresponding to their edgers. We name them \textit{Toom graphs}.

The Toom graph $G$ thus associated to the space-time configuration identifies part of the errors which lead to the negative spin at $(t_-,x_-)$: its set of vertices $V_G$ admits only points where the given configuration presents negative spins and it contains $(t_-,x_-)$. It traces the source of this negative spin, via their propagation according to the deterministic rule conveyed by the function $\phi$, back to some previous errors (when $\phi_x ( \vect{\omega}_{t-1}  ) \allowbreak =+1$ but $\omega_{t,x}=-1$).
Such errors, classified in a distinguished subset $\hat{V}_G$ of $V_G$, all carry a probability smaller than $\epsilon$, because of Assumption (A\ref{assump:2}). $\hat{V}_G$ is noted $S$ in \cite{To80} and constructed iteratively parallel to the truss. The large number of different graphs associated to all possible space-time configurations with a negative spin at $(t_-,x_-)$ is then an `entropic' factor which has to be counterbalanced by an `energetic' factor, i.e. the probability of the graphs, which is bounded by the probability of making all the identified errors, $\epsilon^{|\hat{V}_G|}$. 

Now, using our reformulation of trusses in terms of Toom graphs, the construction algorithm of the trusses in Section \MakeUppercase{\romannumeral 2}, Part $3$ of \cite{To80} guarantees the following properties:
\begin{enumerate}[\hspace{1cm}(P1)]
\item \label{P1}The graphs are connected and contain $(t_-,x_-)$, which is their only vertex at time coordinate $t_-$. We will call it their \textit{origin}. Moreover, their edges are either \textit{arrows}, with a displacement vector between their vertices of the form $\pm (-1, u)$ for a $u$ in $U$, or \textit{forks}, with a displacement vector $(0,u^{(1)}-u^{(2)})$ where $u^{(1)}$ and $u^{(2)}$ are both elements of $U$. Therefore, taking into account the $2^q$ possible allocations of the $q$ poles to the two vertices of an edge, for any given point of $ \{0, \ldots , t_- \} \times \plan$ at most $2^q(R^2+2R)$ different types of edges can have that point as vertex.
For any such connected graph, there always exists a walk which starts from $(t_-,x_-)$, passes along every edge exactly twice and then comes back to its departure point. We choose such a walk and consider the sequence that records, at each of its steps, the displacement vector and the pole distribution of the travelled edge, and call it the \textit{Eulerian walk} associated to the graph. This correspondence is injective. Consequently, the number of different graphs grows only exponentially with their number of edges, $|E_G|$, as the number of Eulerian walks of length $2|E_G|$  is less than $[2^q(R^2+2R)]^{2|E_G|}$. This is Lemma $3$ of \cite{To80}.
\item \label{P2}The number $|E_G|$ of edges is in turn related to the number $|\hat{V}_G|$ of identified errors:
\begin{equation} \label{t:4}
\frac{1}{1+2qr^{-1}} \norm{E_{G}} +1 \leq \norm{\hat{V}_{G}}.
\end{equation}
This is due to a principle which presides over the graph construction in \cite{To80} so as to take advantage of the erosion property: the truss is \textit{overall even} and the number of forks equals the number of identified errors $|\hat{V}_G|$ minus one. Using Lemma $2$ of \cite{To80}, one deduces the upper bound \eqref{t:4}.
\end{enumerate}
Reference \cite{To80} uses these properties to prove that the `energetic' factor wins as long as $\epsilon$ is small enough. The probability of finding a negative spin at a given time and a given site can be made as small as desired by taking $\epsilon$ appropriately small.\\

Here in our treatment of the space-time configurations that present, for all $t$ such that $(t,\gamma_t) \in \gammamoins$, a negative value of the deterministic prescription given by $\phi$ evaluated at time $t-1$ in $\gamma_t+U$, we will hardly need to modify Toom's construction, but we will associate to each space-time configuration a collection of Toom graphs instead of only one, as we are interested in the source of a collection of negative spins, corresponding to the different elements of $\gammamoins$.

We define a map
\begin{align}
g : \mathcal{E} \left( \gammaplus \right) &\to \mathcal{G}  \left( \gammaplus \right) = g \left(\mathcal{E} \left( \gammaplus \right) \right) \label{t:2} \\
\left( \vect{\omega}_t \right) _{t=0}^n &\mapsto G = g \left( \left( \vect{\omega}_t \right) _{t=0}^n\right)\notag
\end{align}
where $G$ is a graph made of a disconnected collection of Toom graphs whose construction, for a given $\left( \vect{\omega}_t \right) _{t=0}^n \in \mathcal{E} \left( \gammaplus \right) $, consists in the following steps.
\begin{enumerate}
\item \begin{sloppypar} We pick $t_1$, the largest $t$ such that $\left( t_1, \gamma_{t_1} \right) \in \gammamoins$. We know that $\phi_{\gamma_{t_1}} \left( \vect{\omega}_{t_1-1}  \right) = -1$.
We consider $\left( \tilde{\vect{\omega}}_t \right) _{t=0}^n$ which is obtained from $ \left( \vect{\omega}_t \right) _{t=0}^n$ by replacing $\omega_{t_1,\gamma_{t_1}}$ with the value $-1$. We construct the Toom graph $G_1$ having the negative spin at $(t_1,\gamma_{t_1})$ as origin and which is associated to $\left( \tilde{\vect{\omega}}_t \right) _{t=0}^n$ according to the algorithm in Section \MakeUppercase{\romannumeral 2}, Part $3$ of \cite{To80}. The whole algorithm remains valid here, the only difference being that no assumption about the initial condition forbids the presence of negative spins at $t=0$ in $\tilde{\vect{\omega}}_0$ while, in \cite{To80}, the initial condition is $\sleb{+}$ and therefore the graph is always contained in $ \{1, \ldots , t_- \} \times \plan$ since the points where the configuration has positive spins can't be vertices. A rule is then added to the algorithm: when the graph under construction reaches negative spins at $t=0$, they are considered equivalent to errors and classed as elements of $\hat{V}_{G_1}$. In other words, this amounts only to a one-unit shift of the initial condition $\sleb{+}$ along the time axis.
\end{sloppypar} 

So we end up with a connected graph $G_1$. Its set of vertices $V_{G_1}$ contains $(t_1,\gamma_{t_1})$ and at all points $(t,x) \in V_{G_1} \setminus \{ (t_1,\gamma_{t_1}) \} $, the space-time configuration $\left( \vect{\omega}_t \right) _{t=0}^n$ takes the value $\omega_{t,x}=-1$. It has a distinguished subset of vertices $\hat{V}_{G_1} \subseteq V_{G_1}$ where the space-time configuration presents errors: 
\begin{equation*}
\forall (t,x) \in \hat{V}_{G_1} \  \textrm{such that} \  t\neq 0, \ \phi _x ( \vect{\omega}_{t-1}  ) \allowbreak =+1.
\end{equation*}
It has a set $E_{G_1}$ of edges connecting its vertices and satisfying \eqref{t:4}.

$G_1$ is the first of the Toom graphs whose union will form the graph $G$.
\item We pick the next maximum $t_2 < t_1$ such that $\left(t_2,\gamma_{t_2} \right) \in \gammamoins$. We perform exactly the same construction as before and associate to $\left( t_2,\gamma_{t_2} \right)$ the graph $G_2$ with properties analogous to those of $G_1$. If $V_{G_1} \cap V_{G_2} \neq \varnothing$ then we discard $G_2$, as we can't count any error twice. Otherwise the union of $G_1$ and $G_2$ will be part of $G$.
\item We repeat the same process up to the lowest time $t$ such that $\left( t,\gamma_t \right) \in \gammamoins$, discarding any Toom graph $G_i$ which intersects one of the previous retained ones.
\item $G$ is defined as the union of all the retained $G_i$, $i \in \{ 1, \dots , | \gammamoins | \}$. $V_G$ is the union of the retained $V_{G_i}$, $\hat{V}_G$ is the union of the retained $\hat{V}_{G_i}$ and $E_G$ is the union of the retained $E_{G_i}$.
\end{enumerate}
The map $g$ of \eqref{t:2} is then completely defined. It is usually not injective, as there are many configurations corresponding to one graph. But $ \mathcal{E} \left( \gammaplus \right)$ can always be written as:
\begin{equation}\label{p:2}
\mathcal{E} \left( \gammaplus \right) = \bigcup_{G \in \mathcal{G}\left( \gammaplus \right)} g^{-1}G.
\end{equation}
Defining, for all $G \in \mathcal{G} \left( \gammaplus \right)$ and for $t=0, \ldots, n$, the subsets $E(G,t) \subseteq X$:
\begin{equation} \label{toom:6}
\charf{E(G,t)}(\vect{\omega}) = \prod_{x: (t,x) \in V_G \setminus \gammamoins} \charf{-}(\omega_x) \prod_{x: (t+1,x) \in \hat{V}_G } \charf{+}\left( \phi _x(\vect{\omega}) \right),
\end{equation}
we have, by construction of the map $g$,
\begin{equation}
\charf{g^{-1}G} \left( \left( \vect{\omega}_t \right) _{t=0}^n \right) \leq \prod_{t=0}^n \charf{E(G,t)}(\vect{\omega}_t). \label{toom:7}
\end{equation}

Let us consider any $G \in \mathcal{G} \left( \gammaplus \right)$. $G$ is the union of $c$ individually connected but pairwise disconnected Toom graphs $(G_{i_1}, \dots , G_{i_c})$. From this point on, they will be noted $(G_{1}, \dots, G_i, \dots , G_{c})$. They all possess Properties (P\ref{P1}) and (P\ref{P2}). We show that $G$ itself inherits similar properties.

First we find an upper bound on the number of different graphs with given numbers of edges and of connected parts.
We define a map $w$ on $\mathcal{G} \left( \gammaplus \right)$. For $G \in \mathcal{G} \left( \gammaplus \right)$, $w(G)$ consists of two elements: first, the list of the origins of $G_1, \dots , G_c$; second, the sequence given by the concatenation of the Eulerian walks associated to $G_1, \dots , G_c$ in Property (P\ref{P1}).
This map $w:\mathcal{G} \left( \gammaplus \right) \to w\left( \mathcal{G} \left( \gammaplus \right) \right)$ is injective.
\begin{proof}
The induction procedure described in Section \MakeUppercase{\romannumeral 2}, Part $3$ of \cite{To80} for the construction of trusses starts from the vertex which we called the origin of the graph. One can check by inspection of the induction steps that the first step creates exactly $q$ arrows leaving from this origin and then no other edge with this vertex will be drawn during the following steps of the construction. Consequently, we know that the Eulerian walk associated to a Toom graph $G_i$ will leave from and come back to the origin of $G_i$ exactly $q$ times. Given any element of the set $w\left( \mathcal{G} \left( \gammaplus \right) \right)$, its unique inverse image can then be deduced from the list of origins and the sequence of steps of the concatenated Eulerian walks by the following method. Starting from the first origin recorded in the list, we add the steps of the sequence and redraw the first connected part of the graph, until the origin has been reached $q$ times. Then we jump to the next origin recorded in the list and read on the sequence of steps until we again come back $q$ times to this second departure point, and so on, until we have read the whole sequence. While redrawing a graph, we give back each travelled edge its pole distribution, which is recorded in the sequence.
\end{proof}
For all $c \in \nat$ and all $m \in \nat$, let $\mathcal{G} \left( \gammaplus, c, m \right)$ be the subset of $\mathcal{G} \left( \gammaplus \right)$ consisting of the graphs with $c$ connected parts and $m$ edges exactly.
Any element of $w\left( \mathcal{G} \left( \gammaplus, c, m \right) \right)$ is made of a list with $c$ origins chosen among the elements of $\gammamoins$ and of a sequence with exactly $2m$ steps.
Therefore Property (P\ref{P1}) is transferred from Toom graphs to graphs $G \in \mathcal{G} \left( \gammaplus, c, m \right)$:
\begin{equation} \label{t:9}
\norm{\mathcal{G} \left( \gammaplus, c, m \right)} \leq {\norm{\gammamoins} \choose c } . [2^q(R^2+2R)]^{2m} \leq 2^{\norm{\gammamoins}}. [2^q(R^2+2R)]^{2m}.
\end{equation}

On the other hand, every connected part verifies Property (P\ref{P2}), that is to say, satisfies the analog of \eqref{t:4}. Summing this inequality over all parts gives a similar property for $G$:
\begin{equation} \label{t:8}
\frac{1}{1+2qr^{-1}} \norm{E_{G}} +c \leq \norm{\hat{V}_{G}}.
\end{equation}

We also have
\begin{equation} \label{t:7}
\norm{\gammamoins} \leq \norm{E_{G}} +c.
\end{equation}
\begin{proof}
It is sufficient to construct an injective map $f:\gammamoins \to E_G \cup \{1, \dots, c\}$. Let us consider any $(t,\gamma_t) \in \gammamoins$. If $(t,\gamma_t)$ belongs to a connected part $G_i$ of $G$ which is contained in $\{0, \dots, t \} \times \plan$, then we know that $(t,\gamma_t)$ is the unique origin of $G_i$. We assign to $(t,\gamma_t)$ the image $f((t,\gamma_t))=i \in \{1, \dots, c\}$. Otherwise, we know that the Toom graph with origin $(t,\gamma_t)$ has been discarded. Then that Toom graph intersects at least one of the retained Toom graphs with origin $(s,\gamma_s)$, $s>t$. Now the discarded graph is contained in $ \{0, \dots , t \} \times \plan$ and the time component of edges of Toom graphs has maximum absolute value $1$, as can be seen in the definition of arrows and forks in Property (P\ref{P1}). Therefore there exists an edge of the retained graph which arrives onto some point $(t,x)$, $ x \in \plan$. Such an edge belongs to $E_G$ and we take it to be the image of $(t,\gamma_t)$ under $f$. The map $f$ thus defined is easily seen to be injective.
\end{proof}
\section{Exponential Convergence to Equilibrium} \label{section:decouplpure}
\begin{sloppypar}
We now combine the collections of Toom graphs introduced in Section \ref{section:graphs} with the paths of influence of Section \ref{section:path} and their pure phase partition of Section \ref{section:pure}. Inserting partition \eqref{p:2} in expansion \eqref{decoupling:2} yields a sum over graphs.
It introduces intricate couplings between the configurations at different times, due to the complexity of the Toom graph construction algorithm in Section \MakeUppercase{\romannumeral 2}, Part 3 of \cite{To80}. Since $\charf{g^{-1}G}$ does not have the property of factorization over time, we use the upper bound \eqref{toom:7} in the equivalent form $\charf{g^{-1}G} \left( \left( \vect{\omega}_t \right) _{t=0}^n \right) = \charf{g^{-1}G} \left( \left( \vect{\omega}_t \right) _{t=0}^n \right) . \prod_{t=0}^n \charf{E(G,t)}(\vect{\omega}_t) $, as indicator functions can only take values $0$ or $1$. Now, for all $\left( \vect{\omega}_t \right) _{t=1}^n$, we have
\begin{align}
&\norm{\proj{\gamma_0} \mu \left[ \tau^{(\gamma_1, \gamma_0 ) } (\vect{\omega}_1,\cdot ) \ \charf{F(\gammaplus,0)\cap E(G,0)} (\cdot) \ \charf{g^{-1}G} \left( \, \cdot \, ; \left( \vect{\omega}_t \right) _{t=1}^{n} \right) \right]} \notag \\
& \qquad \leq \norm{\proj{\gamma_0} \mu} \left[ \norm{\tau^{(\gamma_1, \gamma_0 ) } (\vect{\omega}_1,\cdot )} \charf{F(\gammaplus,0)\cap E(G,0)} (\cdot) \right] , \notag
\end{align}
where $ \norm{\proj{\gamma_0} \mu}$ is not the norm, but the absolute value of the measure $\proj{\gamma_0} \mu$. Consequently, if we define the operator $\tilde{T}^{(y,x)} $ by replacing $\tau^{(y,x)} (\vect{\xi},\vect{\omega})$ with its absolute value $\norm{\tau^{(y,x)} (\vect{\xi},\vect{\omega})}$ in \eqref{path:2}, the graph expansion gives
\begin{align}
&\norm{ \prod_{t=1}^n \left( \charf{F(\gammaplus,t)}   T^{(\gamma_t,\gamma_{t-1})}  \right) \charf{F(\gammaplus,0)}  \proj{\gamma_0} \mu} \label{graph:1} \\
& \qquad \leq \sum_{G \in \mathcal{G} \left( \gammaplus \right)} \norm{ \prod_{t=1}^n \left( \charf{F(\gammaplus,t)\cap E(G,t)}   \tilde{T}^{(\gamma_t,\gamma_{t-1})}  \right) \charf{F(\gammaplus,0)\cap E(G,0)}  \norm{\proj{\gamma_0} \mu}}. \notag
\end{align}
\end{sloppypar}
This expansion will be the starting point of the proof of Theorem \ref{thm:exp}. Before, we need the following lemmas.
\begin{lemma} \label{lemma1}
For all $x \in \plan$, $\gamma \in P_x$, $\gammaplus \subseteq \gamma \setminus \{ \gamma_0\}$ and $G \in \mathcal{G} \left( \gammaplus \right)$, for all $\nu \in \mathcal{M}(X)$ and for all $t \in \{1, \dots , n\}$,
\begin{align*}
&\grandnorm{ \charf{F(\gammaplus,t)\cap E(G,t)}  \tilde{T}^{(\gamma_t,\gamma_{t-1})}  \charf{F(\gammaplus,t-1)\cap E(G,t-1)}  \nu } \notag \\
& \qquad \leq \norm{\charf{F(\gammaplus,t-1)\cap E(G,t-1)} \nu} \ \epsilon^{|\hat{V}_{G,t}|} \ \alpha^{|\gammaplus_t|} \ 2^{|\gammamoins_t|},
\end{align*}
where we introduced the notation $\gammaplus_t= \gammaplus \cap \{(t,x) | x \in \plan \}$ and the analogs $\gammamoins_t$ and $\hat{V}_{G,t}$.
\end{lemma}
\begin{proof}
Using definitions \eqref{path:2} and \eqref{toom:6} of $T^{(\gamma_t, \gamma_{t-1})}$ and $\charf{E(G,t)}$, we already have
\begin{align}
&\grandnorm{ \charf{F(\gammaplus,t)\cap E(G,t)}  \tilde{T}^{(\gamma_t,\gamma_{t-1})}  \charf{F(\gammaplus,t-1)\cap E(G,t-1)}  \nu } \notag \\
&\quad \leq \norm{\charf{F(\gammaplus,t-1)\cap E(G,t-1)} \nu} \cdot \label{together:1}  \\
& \qquad \qquad \cdot \sup_{\vect{\omega} \in F(\gammaplus,t-1)\cap E(G,t-1)} \int \dif \vect{\xi} \ \prod_{x: (t,x) \in V_G \setminus \gammamoins} \charf{-}\left(\xi_{x}\right) \norm{\tau^{(\gamma_t,\gamma_{t-1})} (\vect{\xi}, \vect{\omega} )}. \notag
\end{align}
Now, keeping definition \eqref{path:6} of $\tau^{(\gamma_t,\gamma_{t-1})}$ in mind, the contributions of the spins $\xi_{x}$ at different sites $x \in \plan$ to the integral in \eqref{together:1} are decoupled and factorize.
We first consider the set $\{x \in \plan :(t,x) \in \hat{V}_G \setminus \gamma \} $. Because of Assumption (A\ref{assump:2}), its elements all contribute by a factor bounded by $\epsilon$, since the supremum in \eqref{together:1} is taken over configurations $\vect{\omega}$ in $E(G,t-1)$.
Everywhere else on $ \plan \setminus  \{ \gamma_t \} $, the contribution is trivially bounded by $1$. For $x=  \gamma_t$, we need upper bounds on $\norm{p\left( \xi_{\gamma_t} | \vect{\omega}\right) -p\left( \xi_{\gamma_t} | \vect{\omega}_{\neq \gamma_{t-1}}, a \right) }$:
\begin{itemize}
\item if $\phi _{\gamma_t}( \vect{\omega}) = +1$,
\begin{itemize}
\item and if $ \xi_{\gamma_t}= -1$,
$\norm{p\left( \xi_{\gamma_t} | \vect{\omega}\right) -p\left( \xi_{\gamma_t} | \vect{\omega}_{\neq \gamma_{t-1}}, a \right) } \leq \alpha \ p\left( \xi_{\gamma_t} | \vect{\omega}\right) \leq \alpha \ \epsilon ; $
\item and if $  \xi_{\gamma_t}= +1$,
$\norm{p\left( \xi_{\gamma_t} | \vect{\omega}\right) -p\left( \xi_{\gamma_t} | \vect{\omega}_{\neq \gamma_{t-1}}, a \right) }\leq \alpha \ p\left( \xi_{\gamma_t} | \vect{\omega}\right) , $
\end{itemize}
where we used Assumptions (A\ref{assump:2}) and (A\ref{assump:3});
\item if $\phi _{\gamma_t}( \vect{\omega})= -1$, we use the trivial bound $\norm{p\left( \xi_{\gamma_t} | \vect{\omega}\right) -p\left( \xi_{\gamma_t} | \vect{\omega}_{\neq \gamma_{t-1}}, a \right) } \leq 1$.
\end{itemize}
Therefore we obtain the following three types of upper bounds on the contribution of the spin at $\gamma_t$ to the integral in \eqref{together:1}: the contribution is bounded by $\alpha \ \epsilon$ if $(t,\gamma_t) \in \gammaplus \cap \hat{V}_G$, by $\alpha$ if $(t,\gamma_t) \in \gammaplus \setminus \hat{V}_G$ or by $2$ if $(t,\gamma_t) \in \gammamoins$.
Inserting all these bounds in \eqref{together:1} proves Lemma \ref{lemma1}.
\end{proof}
\begin{lemma} \label{lemma:2}
For any MBT without memory characterized by a non-constant monotonic function $\phi$ verifying the Toom erosion condition, there exist $\alpha_* >0$ and a positive and strictly decreasing function $\epsilon_*$ defined on $[0,\alpha_*[$ such that, for all $\alpha \in [0, \alpha_* [$ and $\epsilon \in [0, \epsilon_* (\alpha)[$, the following result is verified for any PCA satisfying the corresponding Assumptions (A\ref{assump:2}) and (A\ref{assump:3}). Let $K\geq 0$ and $\epsilon' \in [ 0, \epsilon_* (\alpha) [$. Let $\mu \in \mathcal{M}(X)$ be such that its absolute value $|\mu|$ is in the class $\brond{+} ( K , \epsilon  ')$. Then, there exist some constants $C< \infty$ and $\sigma<1$ such that, for all $\varphi \in \mathcal{C}(X)$ and all $n \in \nat$,
\begin{equation*}
\sum_{x \in \plan} \sum_{ \gamma \in P_x}  \sum_{\gammaplus \subseteq \gamma \setminus \{ \gamma_0 \}}  \norm{ \prod_{t=1}^n \left( \charf{F(\gammaplus,t)}   T^{(\gamma_t,\gamma_{t-1})}  \right) \charf{F(\gammaplus,0)}  \proj{\gamma_0} \mu} \norm{\delta_x \varphi}_{\infty} \leq C  \dnorm{\varphi} \sigma^n  .
\end{equation*}
\end{lemma}
\begin{proof}
Multiple uses of Lemma \ref{lemma1} on the RHS of \eqref{graph:1} imply
\begin{align*}
&\norm{ \prod_{t=1}^n \left( \charf{F(\gammaplus,t)\cap E(G,t)}   \tilde{T}^{(\gamma_t,\gamma_{t-1})}  \right) \charf{F(\gammaplus,0)\cap E(G,0)}  \norm{\proj{\gamma_0} \mu}}\notag \\
&\qquad \qquad \leq \epsilon^{\sum_{t=1}^n |\hat{V}_{G,t}|} \ \alpha^{ |\gammaplus|} \ 2^{|\gammamoins|}  \norm{ \charf{F(\gammaplus,0)\cap E(G,0)}  \norm{\proj{\gamma_0} \mu}}.
\end{align*}
The important point is now that we assumed that $|\mu|$ belongs to $\brond{+}(K,\epsilon')$. Indeed, since $\charf{E(G,0)} \leq \charf{-,\hat{V}_{G,0}}$ and $\charf{-,\hat{V}_{G,0}}( \vect{\omega}_{\succeq \gamma_0}, \vect{a}_{\prec \gamma_0} ) \leq \charf{-,\hat{V}_{G,0}} ( \vect{\omega})$, and using the definition of $\proj{\gamma_0}$, we obtain
\begin{equation*}
\norm{ \charf{F(\gammaplus,0)\cap E(G,0)}  \norm{\proj{\gamma_0} \mu}} \leq 2  K \epsilon'^{\norm{\hat{V}_{G,0}}}.
\end{equation*}
Consequently, with $\tilde{\epsilon}= \max \{ \epsilon, \epsilon' \} $,
\begin{equation}
\norm{ \prod_{t=1}^n \left( \charf{F(\gammaplus,t)\cap E(G,t)}   \tilde{T}^{(\gamma_t,\gamma_{t-1})}  \right) \charf{F(\gammaplus,0)\cap E(G,0)}  \norm{\proj{\gamma_0} \mu}} \leq 2 K \  \tilde{\epsilon}^{|\hat{V}_{G}|} \ \alpha^{ |\gammaplus|} \ 2^{|\gammamoins|}. \label{final:2}
\end{equation}

Now for all $G \in \mathcal{G}\left( \gammaplus \right)$ with $c$ connected parts, \eqref{t:7} implies that the number of edges can be written as $\norm{E_G} = \norm{\gammamoins} - c +k$ with a certain $k \in \nat$. Therefore $ \mathcal{G}\left( \gammaplus \right) = \bigcup_{\substack{c \in \nat \\ k \in \nat}} \mathcal{G} \left( \gammaplus , c, \norm{\gammamoins} - c + k \right)$ and, by virtue of the graph properties established above, we have, using first equation \eqref{t:8} and then equation \eqref{t:9},
\begin{align}
\sum_{G \in \mathcal{G} \left( \gammaplus \right)}  \tilde{\epsilon}^{|\hat{V}_{G}|} &\leq \sum_{c=0}^\infty \ \sum_{k=0}^\infty \norm{\mathcal{G} \left( \gammaplus , c , \norm{\gammamoins } -c +k \right)}  \tilde{\epsilon}^{\frac{1}{1+2qr^{-1}} ( | \gammamoins| -c +k)+c}   \notag \\
&\leq \sum_{c=0}^\infty \ \sum_{k=0}^\infty 2^{\norm{\gammamoins}} . [2^q(R^2+2R)]^{2(|\gammamoins|-c+k)} \tilde{\epsilon}^{\frac{1}{1+2qr^{-1}} ( | \gammamoins| +2qr^{-1} c+k)} \notag \\
&\leq \frac{\left( 2.[2^q(R^2+2R)]^2 \tilde{\epsilon}^{\frac{1}{1+2qr^{-1}}} \right) ^{|\gammamoins|} }{\bigg(1-[2^q(R^2+2R)]^2 \tilde{\epsilon}^{\frac{1}{1+2qr^{-1}}}\bigg) \bigg(1-[2^q(R^2+2R)]^{-2} \tilde{\epsilon}^{\frac{2qr^{-1}}{1+2qr^{-1}}}\bigg)} , \label{final:1}
\end{align}
provided that $\epsilon$ and $\epsilon'$ are such that $[2^q(R^2+2R)]^2 \tilde{\epsilon}^{1/(1+2qr^{-1})} <1$.
\begin{sloppypar}
Combining \eqref{graph:1}, \eqref{final:2} and \eqref{final:1}, we find
\begin{align*}
&\sum_{x \in \plan} \sum_{ \gamma \in P_x}  \sum_{\gammaplus \subseteq \gamma \setminus \{ \gamma_0 \}}  \norm{ \prod_{t=1}^n \left( \charf{F(\gammaplus,t)}   T^{(\gamma_t,\gamma_{t-1})}  \right) \charf{F(\gammaplus,0)}  \proj{\gamma_0} \mu} \norm{\delta_x \varphi}_{\infty} \notag \\
& \qquad \qquad \leq C \sum_{x \in \plan} \sum_{ \gamma \in P_x}  \sum_{\gammaplus \subseteq \gamma \setminus \{ \gamma_0 \}}  \alpha^{ |\gammaplus|} \left( 4.[2^q(R^2+2R)]^2 \tilde{\epsilon}^{\frac{1}{1+2qr^{-1}}} \right) ^{|\gammamoins|} \norm{\delta_x \varphi}_{\infty},
\end{align*}
where
\begin{equation*}
C=\frac{2 K}{ \bigg(1-[2^q(R^2+2R)]^2 \tilde{\epsilon}^{\frac{1}{1+2qr^{-1}}} \bigg) \bigg(1-[2^q(R^2+2R)]^{-2} \tilde{\epsilon}^{\frac{2qr^{-1}}{1+2qr^{-1}}}\bigg)}.
\end{equation*}
Here, we can use Newton's binomial formula: for any finite set $A$ and any $x$ and $y$ in $\real$, $\sum_{ B \subseteq A} x^{\norm{B}} y^{\norm{A \setminus B}} = (x+y)^{\norm{A}}$. Finally, since $\norm{P_x} =R^n$ and keeping definition \eqref{re:1} in mind, we get the desired upper bound if we take $\sigma =  R ( \alpha + 4 . [2^q(R^2+2R)]^2 \tilde{\epsilon}^{1/(1+2qr^{-1})}  )$. For $\sigma$ to be lower than $1$, the parameters must satisfy $\alpha < \alpha_* = 1/R$ and $\tilde{\epsilon} = \max \{ \epsilon, \epsilon' \}  < \epsilon_* (\alpha) $ where the function $\epsilon_*$ is strictly decreasing and positive on $[0, \alpha_* [$.\end{sloppypar}
\end{proof}
Using Lemmas \ref{lemma1} and \ref{lemma:2} and the previous sections, we now prove Theorem \ref{thm:exp}.
\begin{proof} [Theorem \ref{thm:exp}]
We take the constant $\alpha_*$ and the function $\epsilon_*$ obtained in Lemma \ref{lemma:2}. Let us consider any PCA satisfying Assumptions (A\ref{assump:2}) and (A\ref{assump:3}) with parameters $\alpha <\alpha_*$ and $\epsilon < \epsilon_*(\alpha)$.
We write $T^n \mu (\varphi) - \muinv{+} (\varphi) = T^n \bar{\mu} (\varphi) $ where $\bar{\mu}= \mu - \muinv{+} \in \mathcal{M}(X)$ is a signed measure of zero mass. Therefore we can apply the above path expansion and pure phase expansion to $T^n \bar{\mu}$ which then satisfies \eqref{decoupling:2}. Using definition \eqref{muinv}, $\bar{\mu}$ can be rewritten as $\bar{\mu} =  \lim_{j \to \infty}^* \frac{1}{n_j} \sum_{k=0}^{n_j-1} \left( \mu - T^k \sleb{+} \right)$, so that $\norm{T^n \bar{\mu}(\varphi) }$ is bounded above by
\begin{align}
& \liminf_{j \to \infty} \frac{1}{n_j} \sum_{k=0}^{n_j-1} \sum_{x \in \plan} \sum_{ \gamma \in P_x} \sum_{\gammaplus \subseteq \gamma \setminus \{ \gamma_0 \}}  \left[ \norm{ \prod_{t=1}^n \left( \charf{F(\gammaplus,t)}   T^{(\gamma_t,\gamma_{t-1})}  \right) \charf{F(\gammaplus,0)} \proj{\gamma_0}  \mu} \right.\notag \\
&\qquad \qquad+ \left. \norm{  \prod_{t=1}^n \left( \charf{F(\gammaplus,t)}   T^{(\gamma_t,\gamma_{t-1})}  \right) \charf{F(\gammaplus,0)} \proj{\gamma_0}  T^k \sleb{+} }\right] \norm{\delta_x \varphi}_{\infty}. \label{exp:1}
\end{align}

As the probability measure $\mu$ belongs to $  \brond{+} \left( K, \epsilon' \right)$ with $\epsilon'<\epsilon_*(\alpha)$, Lemma \ref{lemma:2} applies to the first part of \eqref{exp:1}:
\begin{equation} \label{exp:7}
\sum_{x \in \plan} \sum_{ \gamma  \in P_x} \sum_{\gammaplus \subseteq \gamma \setminus \{ \gamma_0 \}}  \norm{ \prod_{t=1}^n \left( \charf{F(\gammaplus,t)}   T^{(\gamma_t,\gamma_{t-1})}  \right) \charf{F(\gammaplus,0)} \proj{\gamma_0}  \mu} \norm{\delta_x \varphi}_{\infty} \leq C \dnorm{\varphi} \sigma^n,
\end{equation}
where $C<\infty$ and $\sigma<1$ are given in Lemma \ref{lemma:2}. 

As for the second part of \eqref{exp:1}, it can be bounded thanks to a slight modification of the same arguments. We don't know whether all considered models have the property that $T^k \sleb{+}$ belongs to some $\brond{+} \left( K, \epsilon' \right)$.
Therefore, we will now extend the Toom graphs up to time $-k$ instead of $0$. For fixed $k \in \nat$, $x \in \plan$, $\gamma \in P_x$ and $\gammaplus \subseteq  \gamma \setminus \{ \gamma_0 \}$, we define
\begin{equation*}
\mathcal{E}_k \left( \gammaplus \right) = \left\{ \left( \vect{\omega}_t \right) _{t=-k}^n \in X^{\{-k,\ldots , n \}} \mid \vect{\omega}_{t} \in F \left( \gammaplus, t \right) \quad \forall t \in \{ 0, \dots, n \} \right\}.
\end{equation*}
A map $g_k$ is defined, similarly to the map $g$ above:
\begin{align*}
g_k : \mathcal{E}_k \left( \gammaplus \right) &\to \mathcal{G}_k  \left( \gammaplus \right) = g_k \left(\mathcal{E}_k \left( \gammaplus \right) \right)\\
\left( \vect{\omega}_t \right) _{t=-k}^n &\mapsto G = g_k \left( \left( \vect{\omega}_t \right) _{t=-k}^n\right).\notag
\end{align*}
The only change in the graph construction algorithm takes place whenever a Toom graph $G_i$ under construction reaches a negative spin at time $t=0$. Instead of classing it into the set of identified errors $\hat{V}_{G_i}$, we carry on the construction of this branch of $G_i$, as for positive times, until we meet either an error or a negative spin at time $-k$, which is now considered the initial time, and class it in $\hat{V}_{G_i}$. Again, this merely amounts to a translation of the initial condition $\sleb{+}$ by $k$ units along the time axis. At the end, we obtain a disconnected union of Toom graphs on $\{ -k,, \dots, n\} \times \plan$, with origins $\left( t,\gamma_{t} \right) \in \gammamoins$ and with the properties described above. In particular, all the previous results about $\mathcal{G} \left( \gammaplus \right)$ still hold for $\mathcal{G}_k \left( \gammaplus \right)$ and for the associated subsets $\mathcal{G}_k \left( \gammaplus, c, m \right)$ of graphs with $c$ connected parts and $m$ edges.

Extending definition \eqref{toom:6} of $ E(G,t)$ to graphs $G \in \mathcal{G}_k  \left( \gammaplus \right) $ and to negative times, the analogs for $g_k$ of \eqref{p:2} and \eqref{toom:7} lead to the following graph expansion for the second part of \eqref{exp:1}:
\begin{align}
&\norm{ \prod_{t=1}^n \left( \charf{F(\gammaplus,t)}   T^{(\gamma_t,\gamma_{t-1})}  \right) \charf{F(\gammaplus,0)}  \proj{\gamma_0} T^k \sleb{+} } \notag \\
& \qquad \leq \sum_{G \in \mathcal{G}_k \left( \gammaplus \right)} \left\lvert  \prod_{t=1}^n \left( \charf{F(\gammaplus,t)\cap E(G,t)}   \tilde{T}^{(\gamma_t,\gamma_{t-1})}  \right) \right. \label{exp:6} \\
& \qquad \qquad \qquad \qquad \left. \circ \ \charf{F(\gammaplus,0)\cap E(G,0)}  \Theta_{\gamma_{0}} \prod_{t=-k+1}^{-1} \left( \charf{E(G,t)} T \right) \charf{E(G,-k)} \sleb{+} \right\rvert ,\notag
\end{align}
where we defined the operator $\Theta_x$:
\begin{equation*}
\Theta_x \mu (\varphi) = \int \dif \vect{\omega} \int \dif \vect{\xi} \; \varphi(\vect{\xi}) \norm{\left( \proj{x} \prob[ \, \cdot \, | \vect{\omega}] \right) \left( \vect{\xi} \right) } \mu ( \vect{\omega} ) ,
\end{equation*}
with $\proj{x}$ acting on the measure $\prob[\, \cdot  \, | \vect{\omega}]$.

Lemma \ref{lemma1} as well extends immediately to graphs $G \in \mathcal{G}_k  \left( \gammaplus \right) $. Applying it $n$ times to the RHS of \eqref{exp:6}, we find that it is bounded above by
\begin{equation*}
\sum_{G \in \mathcal{G}_k \left( \gammaplus \right)} \epsilon^{\sum_{t=1}^n |\hat{V}_{G,t}|} \alpha^{ |\gammaplus|}  2^{|\gammamoins|} \norm{ \charf{F(\gammaplus,0)\cap E(G,0)}  \Theta_{\gamma_{0}} \prod_{t=-k+1}^{-1} \left( \charf{E(G,t)} T \right) \charf{E(G,-k)} \sleb{+} } .
\end{equation*}
Then we can use again Assumption (A\ref{assump:2}) $k$ times.
Indeed, we know from definition \eqref{toom:6} that $\charf{E(G,t)} \leq \charf{-,\hat{V}_{G,t}}$ and from Assumption (A\ref{assump:2}) that $\norm{\charf{-,\hat{V}_{G,t}} T\charf{E(G,t-1)} \nu }\leq \epsilon^{ |\hat{V}_{G,t}|} \norm{\charf{E(G,t-1)} \nu } $ for any probability measure $\nu$.
The operator $\Theta_{\gamma_{0}} $ can be handled in the same way as $T$, taking its definition into account, together with the fact that $\charf{-,\hat{V}_{G,0}}( \vect{\omega}_{\succeq x}, \vect{a}_{\prec x} ) \leq \charf{-,\hat{V}_{G,0}} ( \vect{\omega})$. It will simply introduce an extra factor of $2$ due to the operator $\proj{\gamma_0}$.
Lastly, by definition of $\sleb{+}$, $\sleb{+} \left( \charf{E(G,-k)} \right)=0$ unless $\hat{V}_{G,-k}$ is empty, in which case $\sleb{+} \left( \charf{E(G,-k)}  \right)=1$. Consequently, the RHS of \eqref{exp:6} is lower than
\begin{equation*}
 2  \sum_{G \in \mathcal{G}_k \left( \gammaplus \right)} \epsilon^{ |\hat{V}_{G}|} \ \alpha^{ |\gammaplus|} \ 2^{ |\gammamoins|}. \label{exp:5}
\end{equation*}
And so, performing the same calculations as for the proof of Lemma \ref{lemma:2}, we obtain, for the second part of \eqref{exp:1},
\begin{align}
& \sum_{x \in \plan}\ \sum_{ \gamma \in P_x} \ \sum_{\gammaplus \subseteq \gamma \setminus \{ \gamma_0 \}}  \norm{ \prod_{t=1}^n \left( \charf{F(\gammaplus,t)}   T^{(\gamma_t,\gamma_{t-1})}  \right) \charf{F(\gammaplus,0)} \proj{\gamma_0} T^k \sleb{+}} \norm{\delta_x \varphi}_{\infty} \notag \\
& \qquad \leq C_{\text{inv}}  \dnorm{\varphi} \sigma^n , \label{exp:8}
\end{align}
where
\begin{equation*}
C_{\text{inv}}=\frac{2 }{ \bigg(1-[2^q(R^2+2R)]^2 \epsilon^{\frac{1}{1+2qr^{-1}}} \bigg) \bigg(1-[2^q(R^2+2R)]^{-2} \epsilon^{\frac{2qr^{-1}}{1+2qr^{-1}}}\bigg)}.
\end{equation*}

Inserting \eqref{exp:7} and \eqref{exp:8} in \eqref{exp:1}, we conclude:
\begin{equation}
\norm{T^n \mu (\varphi) - \muinv{+} (\varphi) }  \leq C \dnorm{\varphi}  \sigma^n , \notag
\end{equation}
where we renamed $C +C_{\text{inv}} $ to $C$ for simplicity.
\end{proof}
\section{Exponential Decay of Correlations}\label{section:corollaries}
We will now prove a well-known consequence of Theorem \ref{thm:exp}: the invariant measure $\muinv{+}$ presents exponential decay of correlations in space.
\begin{sloppypar}
\begin{proof}[Corollary]
Since $\sleb{+}$ belongs to $\brond{+}(1,0)$, Theorem \ref{thm:exp} implies that for some $C < \infty$ and some $\sigma < 1$, we have, for all $k \in \nat$,
\begin{equation*}\begin{cases}
\norm{T^k \sleb{+}(\varphi \psi) - \muinv{+}(\varphi \psi)} \leq C \dnorm{\varphi \psi} \sigma^{k} ; \\
\norm{T^k \sleb{+}(\varphi) - \muinv{+}(\varphi)} \leq C \dnorm{\varphi} \sigma^{k} ; \\
\norm{T^k \sleb{+}(\psi) - \muinv{+}(\psi)} \leq C  \dnorm{\psi} \sigma^{k} . \end{cases}\end{equation*}
But $\sleb{+}$ is a product measure and the interactions are local, so
\begin{equation*}
T^k \sleb{+}(\varphi \psi) = T^k \sleb{+}(\varphi)  T^k \sleb{+}(\psi)
\end{equation*}
as long as $2 k \,  \max_{u \in U} \dnorm{u}_1< d(\varphi,\psi)$, where $\dnorm{\cdot}_1$ is the Manhattan norm. 
Since $d(\varphi,\psi) > 0$, we also have
\begin{equation*}
\dnorm{\varphi \psi } \leq \dnorm{\varphi} \norm{\psi}_{\infty} + \norm{\varphi}_{\infty} \dnorm{\psi} < \infty.
\end{equation*}
Hence, as long as $2 k \,  \max_{u \in U} \dnorm{u}_1< d(\varphi,\psi)$,
\begin{equation*}
\norm{\muinv{+}(\varphi \psi) - \muinv{+}(\varphi) \muinv{+}(\psi)} \leq 2 C \left( \dnorm{\varphi} \norm{\psi}_{\infty} + \norm{\varphi}_{\infty} \dnorm{\psi} \right) \sigma^{k}.
\end{equation*}
If we take $k = \lceil d(\varphi,\psi)/2 \max_{u \in U} \dnorm{u}_1 \rceil -1$ and choose $C'= \frac{2C}{\sigma}$ and $\eta=\sigma^{\frac{1}{2 \max_{u \in U} \dnorm{u}_1}}$, we obtain the desired inequality.
\end{proof}\end{sloppypar}

The invariant measure $\muinv{+}$ also exhibits exponential decay of correlations in time. 
\begin{proof}[Theorem \ref{coroll:2}]
Theorem \ref{thm:exp} applied to $\sleb{+} \in \brond{+}(1,0)$ implies that, for all $k\in \nat$,
\begin{equation} \label{cor:3}
\norm{\muinv{+} (\varphi T^n \psi) - T^k \sleb{+} (\varphi T^n \psi) } \leq C \dnorm{\varphi T^n \psi} \sigma^k . 
\end{equation}
Now, remembering definition \eqref{re:1} of the semi-norm $\dnorm{.}$, we notice that 
\begin{equation*}
\dnorm{\varphi T^n \psi} \leq 2 \norm{\varphi}_{\infty} \norm{\psi}_{\infty} \norm{\text{supp}(\varphi T^n \psi)}.
\end{equation*}
But the conical structure of the space-time influence of spins implies
\begin{equation*}
\norm{\text{supp}(\varphi T^n \psi)}  \leq \norm{\text{supp}(\varphi)} + \big( \text{diam}(\text{supp}(\psi))+2n \max_{u \in U} \dnorm{u}_1 \big)^d,
\end{equation*}
so there exists a finite constant $  C'_{\varphi, \psi}$ such that $\dnorm{\varphi T^n \psi}\leq C'_{\varphi, \psi} \, n^d $ for all $n \in \nat_0$.

Theorem \ref{thm:exp} also yields
\begin{align}
&\norm{\muinv{+}(\varphi) \muinv{+}(\psi) - \muinv{+}(\varphi)T^{k+n} \sleb{+}(\psi)} \leq C \norm{\varphi}_{\infty} \dnorm{\psi} \sigma^{k+n}  ;\label{cor:7} \\
&\norm{\muinv{+}(\varphi)T^{k+n} \sleb{+}(\psi) - T^k \sleb{+}(\varphi ) T^{k+n} \sleb{+}( \psi)} \leq C  \norm{ \psi}_{\infty} \dnorm{\varphi} \sigma^k . \label{cor:8}
\end{align}

In order to find an upper bound for $ \norm{T^k \sleb{+}\big( \left(\varphi - T^k \sleb{+} (\varphi) \right)  T^n\psi\big) }$, we rewrite it as $\norm{T^n \mu^{(k)}_{\varphi} (\psi)}$, where the signed measure $\mu^{(k)}_{\varphi} \in \mathcal{M}(X)$ is defined by
\begin{equation*}
\mu^{(k)}_{\varphi} (\psi) = T^k \sleb{+} \left( \left( \varphi - T^k \sleb{+} (\varphi) \right) \psi \right) \quad \forall \psi \in \mathcal{C}(X) .
\end{equation*}
$\mu^{(k)}_{\varphi} $ is a measure of zero mass so it satisfies \eqref{decoupling:2}:
\begin{equation*}
\norm{T^n \mu^{(k)}_{\varphi}(\psi)} \leq \sum_{x \in \plan}\ \sum_{ \gamma \in P_x} \ \sum_{\gammaplus \subseteq \gamma \setminus \{ \gamma_0 \}}  \norm{\prod_{t=1}^n \left( \charf{F(\gammaplus,t)}   T^{(\gamma_t,\gamma_{t-1})}  \right) \charf{F(\gammaplus,0)}  \proj{\gamma_0}\mu^{(k)}_{\varphi}}\norm{\delta_x \psi}_{\infty} .
\end{equation*}
The last expression is similar to the second term of \eqref{exp:1} with $\mu^{(k)}_{\varphi}$ instead of $T^k \sleb{+}$ and we can then perform the same argument as in the proof of Theorem \ref{thm:exp} in order to bound $\norm{T^n \mu^{(k)}_{\varphi} (\psi)}$, extending again the collections of Toom graphs up to time $-k$. The only change in these calculations is the presence of an extra factor $\varphi- T^k \sleb{+} (\varphi)$ whose supremum norm is bounded by $2 \norm{\varphi}_{\infty}$. 
Provided we keep track of this factor, the calculations which lead to \eqref{exp:8} still hold here:
\begin{equation}\label{cor:6}
 \norm{T^k \sleb{+}\left( \left(\varphi - T^k \sleb{+} (\varphi) \right)  T^n\psi\right) } \leq 2  \  C_{\text{inv}}  \norm{\varphi}_{\infty} \dnorm{\psi} \sigma^n .
\end{equation}
Combining \eqref{cor:3}, \eqref{cor:7}, \eqref{cor:8} and \eqref{cor:6} and taking the $k \to \infty$ limit ends the proof.
\end{proof}

Theorem \ref{coroll:2} implies that $\muinv{+}$ is not only ergodic but also mixing, that is, for any continuous functions $\varphi, \psi$, we have
\begin{equation*}
\lim_{n \to \infty} \norm{\muinv{+}(\varphi T^n\psi) - \muinv{+}(\varphi) \muinv{+}(\psi)}  = 0.
\end{equation*}
Indeed, by definition of $\mathcal{C}(X)$, the set of continuous functions with finite support is dense in $\mathcal{C}(X)$.

\begin{acknowledgements}
The authors would like to thank Jean Bricmont, Carlangelo Liverani and Christian Maes for helpful comments and discussions. They also thank the referees for their constructive remarks and suggestions.\\
Augustin de Maere was partially supported by the Belgian IAP (Interuniversity Attraction Pole) program P6/02.\\Lise Ponselet was supported by a grant from the Belgian F.R.S.-FNRS (Fonds de la Recherche Scientifique) as `Aspirant FNRS'.\\
This article is published in J. Stat. Phys. 147(3), 634-652 (2012). The final publication is available at www.springerlink.com. DOI 10.1007/s10955-012-0487-9
\end{acknowledgements}

\bibliography{bibliocompl}{}
\bibliographystyle{spmpsci} 

\end{document}